\pdfoutput=1
\documentclass[11pt]{article}

\usepackage{amsmath,amsfonts,amssymb,amsthm,graphicx}
\usepackage{xcolor}
\usepackage{hyperref}
\usepackage{paralist}

\definecolor{darkblue}{rgb}{0,0,0.35}
\definecolor{darkred}{rgb}{0.6,0,0}
\definecolor{darkgreen}{rgb}{0.1,0.35,0}
\hypersetup{colorlinks, linkcolor=darkblue, citecolor=darkgreen, urlcolor=darkblue}
\hypersetup{pdfauthor={Neil Olver},pdftitle={A note on hierarchical hubbing for a generalization of the VPN problem}}

\newcommand{\Rplus}{\mathbb{R}_+}

\newtheorem{theorem}{Theorem}
\newtheorem{lemma}[theorem]{Lemma}

\newtheorem*{gvpnconjecture}{Generalized VPN Conjecture}
\newtheorem*{gvpnconjecturetwo}{Generalized VPN Conjecture (II)}
\theoremstyle{definition}
\newtheorem{definition}[theorem]{Definition}
\newtheorem{remark}[theorem]{Remark}
\newtheorem*{rndproblem}{RND problem}
\newtheorem*{rndhhproblem}{\rndhh{} problem}
\newcommand{\emphdef}[1]{\textbf{#1}}
\newcommand{\universe}{\mathcal{U}}
\newcommand{\rndhh}{$\text{RND}_{\text{HH}}$}

\newcommand{\Pload}[2]{\ell(#1, #2)}

\renewcommand{\phi}{\varphi}
\let\oldvec\vec
\renewcommand{\vec}[1]{\boldsymbol{#1}}

\newcommand{\hmap}{\phi}
\newcommand{\emap}{\phi}

\newcommand{\Tb}{T^b}
\newcommand{\captree}[2]{#1^{#2}}
\newcommand{\hubtree}{hub tree}

\DeclareMathOperator{\fundamentalcut}{fund}
\newcommand{\fundcut}[2]{\fundamentalcut_{#1}(#2)}
\newcommand{\ind}[1]{[#1]}

\newcommand{\hh}{hubbing}
\newcommand{\treeUb}[2]{\kern0.1em\mathcal{U}(\captree{#1}{#2})}
\newcommand{\oriented}[2]{\oldvec{#1}_{#2}}

\begin{document}

\title{A note on hierarchical hubbing for a generalization of the VPN problem}

\author{Neil Olver\thanks{Department of Econometrics \& Operations Research, VU University Amsterdam, and Centrum Wiskunde \& Informatica (CWI), Netherlands. Email: \href{mailto:n.olver@vu.nl}{n.olver@vu.nl}. Research supported by a NWO Veni grant.}}

\maketitle

\begin{abstract}
    Robust network design refers to a class of optimization problems that occur when designing networks to efficiently handle variable demands.
    The notion of ``hierarchical hubbing'' was introduced (in the narrow context of a specific robust network design question), by Olver and Shepherd~\cite{OS2010}.
Hierarchical hubbing allows for routings with a  multiplicity of ``hubs'' which are connected to the terminals and to each other in a treelike fashion.
   Recently, Fr\'echette et al.~\cite{INFOCOM} explored this notion much more generally, focusing on its applicability to an extension of the well-studied hose model that allows for upper bounds on individual point-to-point demands. 
In this paper, we consider hierarchical hubbing in the context of a different (and extremely natural) generalization of the hose model, studied earlier in~\cite{OS2010}, and prove that the optimal hierarchical hubbing solution can be found efficiently. 
This result is relevant to a recently proposed generalization of the ``VPN Conjecture''.
\end{abstract}

\section{Introduction}
\subsection{Robust network design}
\emph{Robust network design} considers the problem of building networks under uncertainty in the pattern of utilization.
Introduced by Ben-Ameur and Kerivin~\cite{benameur03}, the framework encompasses the important case of the ``hose model'' introduced by Fingerhut~\cite{fingerhut97} and Duffield et al.~\cite{duffield99}.
It can itself be seen as falling under the broader umbrella of robust optimization~\cite{bental2009robust}.

We refer the reader to~\cite{myphdthesis} for a more in-depth treatment; here we will give a brief self-contained exposition of the model.
We are given an undirected graph $G = (V,E)$; this should be interpreted as an existing high-capacity network, in which we can reserve capacity.
We assume there is an unlimited total capacity on any given link of the network, and that the cost to reserve capacity on any link is a linear function of the capacity required.
Let $c: E \to \Rplus$ denote the per-unit cost of capacity on each edge.
A set $W \subseteq V$ of \emph{terminals} need to be adequately connected using the capacity reserved.

A \emph{traffic pattern} (or \emph{demand pattern}) describes the precise pairwise demand requirements at some moment in time.
It can be specified by a traffic matrix $D$, indexed by pairs of terminals; for terminals $i, j$, the entry $D_{ij}$ represents the bandwidth needed to send data from $i$ to $j$.
In our network, the traffic pattern is not fixed, but varying (and possibly uncertain).
To deal with this, the robust network design framework allows for a \emph{set} of traffic patterns to be prescribed.
This (it turns out) can always be taken to be a convex set, and so we describe this set, or \emph{demand universe}, as a convex body $\universe \subset \Rplus^{W \times W}$.

The robust network design (RND) problem asks for the cheapest capacity reservation $u: E \to \Rplus$ that can support all traffic patterns in the specified universe $\universe$.
To fully specify the problem however, a further aspect must be considered: the \emph{routing scheme}.
The coarsest division is into \emph{oblivious} or \emph{dynamic} routing.
In dynamic routing, the way in which traffic is routed may vary arbitrarily as a function of the current traffic pattern.
This is typically infeasible, and we will be concerned here with the more practical oblivious routing, where the routing used for any given pair of terminals is specified in advance.
We will also only consider \emph{single-path routing}.
The routing scheme in this case is described by a template $\mathcal{P} = \{P_{ij}: i, j \in W\}$, where $P_{ij}$ is an $i$-$j$-path for each $i, j \in W$.
(We do not require this path to be simple.)

We may summarize the general robust network design problem (with oblivious, single-path routing) as follows:

\begin{rndproblem}
Given an undirected graph $G=(V,E)$ with edge costs $c(e)$, a terminal set $W \subseteq V$, and a convex demand universe $\universe \subset \Rplus^{{W \times W}}$, a solution to the robust network design problem consists of a routing template $\mathcal{P} = \{ P_{ij}: i,j \in W\}$, and a capacity allocation $u: E \to \Rplus$, such that $\universe$ can be routed according to $\mathcal{P}$ within the capacity $u$, i.e., 
\begin{equation}\label{eq:u}
    u(e) \geq \max_{D \in \universe} \sum_{i,j \in W} D_{ij}\Pload{P_{ij}}{e} \qquad \forall e \in E.
\end{equation}
Here, $\Pload{P}{e}$ gives the number of times that edge $e$ occurs on the (possibly non-simple) path $P$.
\end{rndproblem}

The difficulty in this optimization problem lies in choosing the routing template; once this is fixed, the optimal capacity allocation can be determined by solving a convex program described by \eqref{eq:u}, assuming we have access to at least a separation oracle for $\universe$.

Note that there is always an optimum solution template whose paths $P_{ij}$ are all simple, since any non-simple path can simply be replaced by a simple path within its support.
The reason we allow non-simple paths is related to the specific type of routing templates we will be interested in.

A case of interest is that of \emph{symmetric} demands; this means that demand from $i$ to $j$ is not distinguishable from demand from $j$ to $i$.
In this case, which will concern us in this paper, it is convenient to consider $\universe$ to be a subset of $\binom{W}{2}$, the set of unordered pairs of terminals, so that $D_{ij} = D_{ji}$ refers to the same demand, and $P_{ij}=P_{ji}$ to the same path. Equation~\eqref{eq:u} then becomes
\begin{equation}\label{eq:u2}    u(e) \geq \max_{D \in \universe} \sum_{\{i,j\} \subseteq W} D_{ij}\Pload{P_{ij}}{e}. \end{equation}

The well-studied \emph{symmetric hose model}~\cite{fingerhut97, duffield99} is parameterized by a vector $b \in \Rplus^{W}$, yielding the universe
\[ \mathcal{H}(b) = \Bigl\{ D \in \Rplus^{\binom{W}{2}} : \sum_{\{i,j\} \subset W} D_{ij} \leq b_i \quad \forall i \in W\Bigr\}. \]
This models the situation where terminals are connected to the network with ``hoses'' of known, fixed capacity, so that the total demand involving terminal $i$ cannot exceed the capacity $b_i$ of its associated hose link.
Any demand pattern that fits through the hoses should be routable in the final network.
These hoses may model real links, or chosen based on operational criteria; either way, the hose model gives a simple, useful and concise description of what the network must be able to handle, making it a very popular model in the literature. 

A number of variations and generalizations of the symmetric hose model have been considered~\cite{eisenbrand2006provisioning, INFOCOM, OS2010, fingerhut97}.
For example, Fr\'echette et al.~\cite{INFOCOM} consider the ``capped'' hose model, where in adddition to the hose capacities $b$, point-to-point upper bounds $\Gamma \in \Rplus^{\binom{W}{2}}$ are also given, and the universe is 
\[ \mathcal{H}(b) \,\cap\, \Bigl\{ D \in \Rplus^{\binom{W}{2}} : D_{ij} \leq \Gamma_{ij} \quad \forall i,j \in W\Bigr\}. \]

\subsection{The generalized VPN problem}\label{sec:gvpn}
Rather then the capped hose model, we will be concerned with a different generalization of the hose model, introduced by Olver and Shepherd~\cite{OS2010}.
Let $T^b$ be an arbitrary capacitated tree, with nonnegative edge capacities $b$ and with leaf set in exact correspondence with the terminal set $W$.
We will use $T^b$ to define a demand universe in a simple and natural way: let $\treeUb{T}{b}$ consist of all demand patterns that can be routed on $T^b$.

The case where $T^b$ is a star corresponds precisely to the hose model; the capacity of the edge adjacent to terminal $i$ precisely gives the marginal of $i$.
This generalization allows the network operator more precise control over the demand universe, hopefully leading to more efficient solutions. 
In particular, if the terminals of the network can be logically divided into distinct groups (e.g., different branches of the company), with limited communication between groups, this information can be encoded via $\treeUb{T}{b}$.

It was shown in~\cite{OS2010} that the RND problem with oblivious routing for this class of demand universes is approximable to within a factor of $8$.
The algorithm that achieves this always returns a solution of a particular form---a \emph{hierarchical hubbing}.

\subsection{Hierarchical hubbing}

Olver and Shepherd~\cite{OS2010} proposed the following algorithm for the RND problem with universe $\treeUb{T}{b}$: find the cheapest embedding of $T^b$ into the network.
This embedding maps internal nodes of $T^b$ into the network (these are ``hubs''), and each edge $e$ of $T^b$ is mapped to a ``cable'' of capacity $b(e)$ that connects the hubs corresponding to the endpoints of the edge (see Figure~\ref{fig:embedding}).
More than one node of $T^b$ can be mapped to the same hub, and multiple cables may run over the same edge of the network.
The routing template associated with a hierarchical hubbing is obtained, for each $\{i,j\} \subseteq W$, as the image of the unique $i$-$j$-path in the tree under the mapping (again, see Figure~\ref{fig:embedding}), yielding a (possibly non-simple) $i$-$j$-path in $G$.
For any edge $e \in E(G)$, define the capacity $u(e)$to be the sum of the cable capacities that use edge $e$. 
It is easy to show that this template and capacity allocation provides a valid solution to the RND problem with universe $\treeUb{T}{b}$.

\begin{figure}
    \centering
    \includegraphics{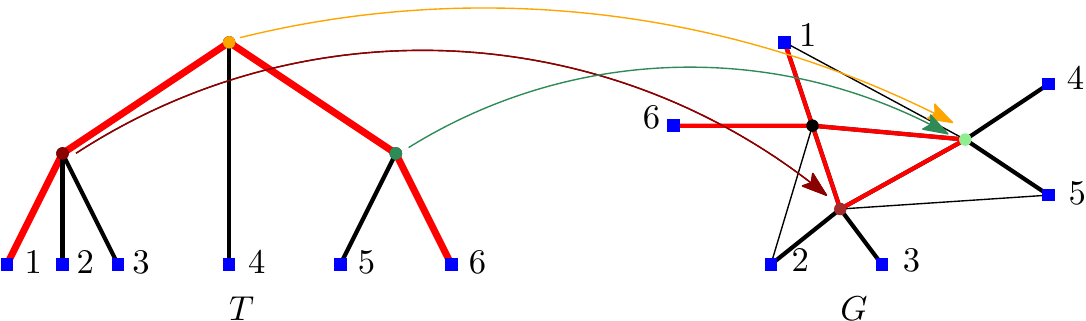}
    \caption{An example of an embedding of $T$, with the resulting routing from $1$ to $4$ indicated.}\label{fig:embedding}
\end{figure}

What is very pleasant about this restricted class of solutions is that the cheapest such solution can be found in polynomial time. (The optimization problem is closely related to the zero-extension problem on trees~\cite{karzanov1998metrics,calinescu2005approximation}.)
It is shown in~\cite{OS2010} that the resulting solution is within a constant factor of the optimal solution to the RND problem for $\treeUb{T}{b}$, where arbitrary routing templates are allowed.

\medskip

Fr\'echette et al.~\cite{INFOCOM} generalize this notion in two ways.
Firstly, they consider hierarchical hubbing templates for arbitrary demand universes, not just those associated with the generalized VPN problem.
Secondly, they consider the optimization problem where the tree used to define the hierarchical hubbing is not specified, but may be chosen as part of the solution.
More precisely, they define the following optimization problem:
\begin{rndhhproblem}
Given an undirected graph $G = (V,E)$ with edge costs $c(e)$, a terminal set $W \subseteq V$, and a convex demand universe $\universe \subset \Rplus^{{W \times W}}$, a solution to the \rndhh{} problem consists of \begin{inparaenum}[(i)] \item  a capacitated tree $T^b$, with leaf set $W$, such that any demand in $\universe$ can be routed on $T^b$, and \item an embedding of $T^b$. \end{inparaenum}
The capacity allocation $u(e)$ for an edge $e \in E$ is the sum of the capacities $b(f)$ of the cables that use edge $e$, and the goal is to minimize the final cost $\sum_{e \in E} c(e)u(e)$ of the solution.
\end{rndhhproblem}
It is easy to confirm that any solution to the \rndhh{} problem is a solution to the RND problem, but not vice versa.
So in general the optimal solution to \rndhh{} can be more expensive than the optimal RND solution; in fact, Fr\'echette et al.~\cite{INFOCOM} demonstrate that the gap can be $\Omega(\log |V|)$, for some choices of the universe.

Fr\'echette et al.~\cite{INFOCOM} are motivated to consider hierarchical hubbing for a few reasons.
In \emph{hub routing}, all traffic is routed via a single hub node; this has the advantage that routing decisions are localized at the hub.
In order to address some practical shortfalls of hub routing, Shepherd and Winzer~\cite{shepherdrlb06} ask for a ``multihub'' extension of this.
Fr\'echette et al.\ argue that hierarchical hubbing provides a natural such extension (note that it is clearly a generalization; hub routing corresponds to taking the hub tree to be a star).
They also show that it provides an effective heuristic for finding good solutions to the capped hose model mentioned earlier, which is APX-hard in the single-path oblivious routing model.
They observe that for the capped hose model (and for other universes as well), hierarchical hubbing can yield much cheaper solutions than using a single hub.
This is in constrast to the vanilla hose model, where a hub routing provides an optimal oblivious routing solution (see Section~\ref{sec:vpnconj}).

Finding the optimal hierarchical hubbing solution for a given universe is in general APX-hard, as observed in~\cite{AlexMSc} (cf.~\cite{kumar2002algorithms} for tree routings); this can be seen by choosing the universe so that the resulting RND problem is precisely the Steiner tree problem.
It is thus natural to ask for which universes the problem is polynomially solvable.
In fact, as essentially shown in~\cite{OS2010} (see Lemma~\ref{lem:dynamic} later), the hierarchical hubbing problem is solvable exactly in polynomial time if the choice of hub tree is specified.  
The difficulty is thus in identifying the correct optimal choice of the hub tree.
One case where the problem was previously known to be polynomial was for the hose model~\cite{guptavpn01}; this will be discussed more in Section~\ref{sec:vpnconj}.

\begin{remark}\label{rem:tree}
    Any tree routing---meaning a routing template $\mathcal{P} = \{ P_{ij}: i,j \in W\}$ such that $\bigcup_{i,j \in W} P_{ij}$ is a tree---can be described as a hierarchical hubbing.
    The hub tree is obtained from the support, adding additional edges as needed to ensure that all terminals are leaves in the hub tree.
\end{remark}
\begin{remark}
    It would also be natural to instead choose capacities by considering the routing template induced by the hierarchical hubbing, and using \eqref{eq:u2}.
    This alternative formulation is in general not the same as described above; there may be situations where not all cables on a given edge can be simultaneously saturated by a traffic pattern in $\mathcal{U}$, leading to a larger capacity requirement with the cable formulation.
    The formulation that we use in this paper, and which is also used in~\cite{INFOCOM}, seems overall easier to deal with (see, e.g., Lemma~\ref{lem:dynamic} later). 
    If the Generalized VPN Conjecture discussed in Section~\ref{sec:vpnconj} is true, it follows immediately that for the universe $\treeUb{T}{b}$, both formulations have a common optimal solution.
\end{remark}

\section{Hierarchical hubbing for the generalized VPN problem}

The main contribution of this note is
\begin{theorem}\label{thm:main}
    The optimal hierarchical hubbing solution for a generalized VPN problem can be found in polynomial time.
\end{theorem}
Thus \rndhh{} is polynomially solvable for a large and interesting class of demand universes.
A further motivation for this result is its connection with a conjecture about the polynomial solvability of the generalized VPN problem with abitrary oblivious single path routings; this is discussed in detail in Section~\ref{sec:vpnconj}.

\subsection{Preliminaries}
We will use the Iverson bracket $[A]$ to denote the indicator function of a predicate $A$.
Given two disjoint subsets $A, B$ of vertices in some graph $G = (V,E)$, an $(A, B)$ cut is any set $S \subseteq V$ that separates $A$ from $B$.
Given a tree $T$ and edge capacities $b: E(T) \to \Rplus$, $\Tb$ will denote the resulting edge-capacitated tree.

\begin{definition}
    We say a tree $T$ is a \emphdef{\hubtree} if its leaf set is in correspondence with the terminal set $W$ of the instance.
\end{definition}

\begin{definition}
    For any \hubtree{} $T$, a \emphdef{$T$-embedding} (into $G$) is a map $\phi: V(T) \cup E(T) \to V(G) \cup E(G)$ such that
    \begin{enumerate}[(i)]
        \item $\phi(v) \in V(G)$ for all $v \in V(T)$,
        \item $\phi(i) = i$ for all $i \in W$, and
        \item $\phi(vw)$ is a simple $\phi(v)$-$\phi(w)$ path in $G$ for each $vw \in E(T)$.
    \end{enumerate}

      \end{definition}
The restriction to simple paths in the above definition is not necessary, but will be notationally convenient; in any case, there is no advantage to using non-simple paths. But again note that for two terminals $i,j \in W$, the $i$-$j$-path induced by $\phi$, obtained by considering the path $iv_1v_2\cdots v_{t-1}j$ in $T$ and concatenating the paths $\phi(iv_1), \phi(v_1v_2), \cdots, \phi(v_{t-1}i)$, may still be non-simple.

\begin{definition}\label{def:hhuniverse}
    A \emphdef{$T$-hubbing} for $\universe$ (into $G$) is a pair $(\phi, u)$ where $\phi$ is a $T$-embedding and $u$ is a corresponding capacity allocation satisfying
    \[ u(e) \geq \sum_{f \in E(T): e \in \emap(f)} b(f), \]
    where for each $f \in E(T)$, $b(f)$ is the maximum load on $f$ induced by some demand in $\universe$.
 (In reference to the earlier discussion, $b(f)$ is the capacity of the ``cable'' associated with edge $f$.)

    A \emphdef{hierarchical hubbing for $\universe$} is simply a $T$-\hh{} for $\universe$ for some choice of \hubtree{} $T$.

    The cost of a hierarchical hubbing solution is defined simply as the cost of the associated capacity allocation.
    \end{definition}

For a given \hubtree{} $T$, it is possible that $\treeUb{T}{b} = \treeUb{T}{b'}$ for distinct capacities $b, b'$.
The following definition is convenient:
\begin{definition}\label{def:defining}
    Edge capacities $b$ for a tree $T$ are called \emphdef{defining} if for every $e \in E(T)$, there exists a $D \in \treeUb{T}{b}$ that saturates edge $e$ in $T^b$.
\end{definition}
Given $T^{b'}$, defining capacities $b$ such that $\treeUb{T}{b} = \treeUb{T}{b'}$ can easily be found, solving one maximum flow problem per edge of $T$.

\subsection{Main theorem}

We will prove (recalling the definition of $\treeUb{T}{b}$ from Section~\ref{sec:gvpn})
\begin{theorem}\label{thm:opthh}
        For a given hub tree $T$ and capacities $b$, the optimal hierarchical hubbing for $\treeUb{T}{b}$ is a $T$-hubbing.
\end{theorem}
We may assume in what follows that the capacities $b$ are defining, since replacing the capacities by defining capacities can be done as a preprocessing step.
Combining this theorem with the following simple but key algorithmic result about hierarchical hubbing, Theorem~\ref{thm:main} immediately follows.
\begin{lemma}[\cite{OS2010}]\label{lem:dynamic}
    For any \hubtree{} $T$ with edge capacities $b$, the optimal $T$-hubbing for $\treeUb{T}{b}$ can be found in polynomial time.
    \end{lemma}
The algorithm is a simple dynamic programming one which finds the optimal placement of the hubs.
Note that in any optimal solution, the path $\phi(vw)$ used to route some $vw \in E(T)$ can always be taken to be a shortest path between $\phi(v)$ and $\phi(w)$.

Theorem~\ref{thm:opthh} is certainly very natural, and it might seem even trivial at first glance.
But consider, for example, the case where $T$ is simply a star, and hence represents a hose model universe.
Then recalling Remark~\ref{rem:tree}---any tree routing is a hierarchical hubbing---this theorem includes the fact that the optimal tree routing under the hose model is a hub routing, a result of Gupta et al.~\cite{guptavpn01}.
Further discussion on this, and connections with the ``VPN Conjecture'' may be found in Section~\ref{sec:vpnconj}.

\medskip

We will in fact prove something stronger than Theorem~\ref{thm:opthh}: given any hierarchical hubbing for $\treeUb{T}{b}$ with capacity allocation $u$, there is a $T$-hubbing for $\treeUb{T}{b}$ which requires no more than capacity $u(e)$ on each edge $e$.
We begin by proving this result for the case where the network is itself a tree, with terminals forming the leaf set of the tree.
We will then observe that the result for an arbitrary network follows easily.

\paragraph*{Tree networks}
Let $F$ be an arbitrary tree, with $W$ denoting the leaves of $F$.
Let $\fundcut{F}{e}$ denote the set of all pairs of leaves that are separated by edge $e \in E(F)$.
Notice that there is an obviously optimal oblivious routing solution for the RND problem with universe $\treeUb{T}{b}$ in $F$: use the template $\mathcal{P} = \{ P_{ij}: i,j \in W\}$ where $P_{ij}$ is the unique simple path from $i$ to $j$ in $F$.
The minimal required capacity $q^*(e)$ for any edge $e \in E(F)$ is then, from \eqref{eq:u2},
\begin{equation}\label{eq:q} q^*(e) =  \max_{D \in \treeUb{T}{b}}\sum_{\{i,j\} \in \fundcut{F}{e}} D_{ij}. \end{equation}
Note that \emph{any} oblivious routing template will need capacity $q^*(e)$ at least on each $e \in E(F)$, since any $i$-$j$ path in $F$ contains $P_{ij}$.

Observe also that this oblivious routing template is a $F$-\hh{} template, induced by the trivial $F$-embedding that maps $F$ to itself.
Since each cable in this embedding is a single edge, combining this $F$-embedding with the capacity allocation $q^*$ yields an $F$-hubbing for $\treeUb{T}{b}$.

The next theorem, which is the key technical theorem of this paper, shows that there is a $T$-hubbing for $\treeUb{T}{b}$ with capacity allocation $q^*$.
In general, the $T$-embedding that certifies this will induce \emph{non-simple} paths; despite this, extra capacity will not be needed.

\begin{theorem}\label{thm:opthhtree}
    Let $F$ be any tree with leaf set $W$, and let $q^*: E(F) \to \Rplus$ be as in \eqref{eq:q}.
    Then there is a $T$-hubbing for $\treeUb{T}{b}$ into $F$ with capacity allocation $q^*$.
\end{theorem}
\begin{proof}

We will need the following lemma, which follows easily from standard uncrossing techniques.
\begin{lemma}\label{lem:uncrossing} Let $W_1 \subseteq W_2 \subseteq W$, and for $i \in \{1,2\}$, let $S_i$ be a minimum capacity $(W_i, W \setminus W_i)$-cut in $T^b$ (with $W_i \subseteq S_i$).
   Then $S_1 \cap S_2$ is a minimum capacity $(W_1, W \setminus W_1)$-cut, and $S_1 \cup S_2$ is a minimum capacity $(W_2, W \setminus W_2)$-cut.
\end{lemma}
\begin{proof}
    By submodularity of the cut function, 
    \begin{equation}\label{eq:uncrossing} b(\delta(S_1 \cap S_2)) + b(\delta(S_1 \cup S_2)) \leq b(\delta(S_1)) + b(\delta(S_2)). 
    \end{equation}
    But $S_1 \cap S_2$ is a $(W_1, W \setminus W_1)$-cut, and so by the definition of $S_1$, 
    $b(\delta(S_1 \cap S_2)) \geq b(\delta(S_1))$.
    Moreover, $S_1 \cup S_2$ is a $(W_2, W \setminus W_2)$-cut, and hence $b(\delta(S_1 \cup S_2)) \geq b(\delta(S_2)$.
    We deduce that \eqref{eq:uncrossing} holds with equality, and hence that $S_1 \cap S_2$ is a minimum $(W_1, W \setminus W_1)$-cut, and $S_1 \cup S_2$ a minimum $(W_2, W \setminus W_2)$-cut.
\end{proof}

    Pick an arbitrary leaf $r \in W$, and call it the root.
An edge $e$ of $F$ divides the terminal set into two, $W_e$ and $W \setminus W_e$, where we choose $W_e$ to \emph{not} contain the root. 
Let $S_e$ be a minimum $(W_e, W \setminus W_e)$ cut containing $W_e$ in the tree $T^b$, breaking ties by choosing $S_e$ to have minimum cardinality.

We now describe the $T$-embedding $\hmap$ into $F$.
We will define, for each internal node $w \in V(T) \setminus W$, an orientation $\oriented{F}{w}$ of $F$.
For each edge $e \in E(F)$, orient $e$ away from the root if $w \in S_e$, and otherwise orient $e$ towards the root.
We then define $\hmap(w)$ to be the unique sink node of $\oriented{F}{w}$, whose existence we guarantee as follows.
\begin{lemma}
    There is a unique sink node in $\oriented{F}{w}$.
   \end{lemma}
\begin{proof}
    We begin by showing that every node has outdegree at most $1$ in $\oriented{F}{w}$.
    Suppose for a contradiction that some node $u \in V(F)$ has two outgoing arcs $e, e'$ in $\oriented{F}{w}$.
    There are two cases to consider:
    \begin{enumerate}[(i)]
        \item Both arcs $e$ and $e'$ are oriented away from $r$. Then $w \in S_e$ and $w \in S_{e'}$.
            Now $W_{e} \subseteq (W \setminus W_{e'})$, and hence by Lemma~\ref{lem:uncrossing}, $\hat{S}_e := S_e \cap (V(T) \setminus S_{e'})$ is a minimum $(W_e, W \setminus W_e)$-cut in $T^b$. But since $w \notin \hat{S}_e$, $\hat{S}_e \subsetneq S_e$, contradicting the size minimality of $S_e$.
        \item One of $e, e'$ is oriented towards $r$, say $e'$. Then $w \in S_e$ and $w \notin S_{e'}$.
            Then simply note that $W_e \subseteq W_{e'}$, and hence (by Lemma~\ref{lem:uncrossing}) $S_e \cap S_{e'}$ is a minimum $(W_e, W \setminus W_e)$-cut in $T^b$. Again since $w \notin S_e \cap S_{e'}$, this contradicts the size minimality of $S_e$.
    \end{enumerate}
    By starting at an arbitrary node, and following the unique outgoing arc until we reach a node with no outgoing arcs (this may be a leaf, or not), the existence of some sink $v$ follows.
    To see that this sink is unique, observe that on any path in $\oriented{F}{w}$ terminating at $v$, all arcs must be oriented towards $v$ by the condition on the outdegree, and hence none of the nodes on the path aside from $v$ can possibly be a sink.
\end{proof}
We complete the definition of the $T$-embedding $\phi$ in the obvious way, by taking $\emap(vw)$ to be the unique simple path between $\hmap(v)$ and $\hmap(w)$ in $F$, for each $vw \in E(T)$. 

\medskip

Now let us consider $q^*(e)$ for some edge $e \in E(F)$.
Rewriting \eqref{eq:q}, we have
\begin{equation}\label{eq:mfmc} q^*(e) = \max_{D \in \treeUb{T}{b}}\sum_{i \in W_e, j \notin W_e} D_{ij}. \end{equation}
The right hand side of \eqref{eq:mfmc} can be seen as a maximum flow problem; send as much flow as possible in $T^b$ from $W_e$ to $W \setminus W_e$.
Invoking the max-flow min-cut theorem, we obtain 
\begin{equation}\label{eq:ucut}
    q^*(e) = b(\delta(S_e)).
\end{equation}

We show now that the $T$-embedding $\phi$ along with the capacity allocation $q^*$ defines a valid hierarchical hubbing solution.
Recall that $[\,\cdot\,]$ denotes the Iverson bracket.
Consider any edge $e$ in $F$. 
The capacity required on edge $e$ by the hierarchical hubbing solution induced by $\phi$ is
\begin{align*} &\phantom{=}\sum_{f=uv \in E(T)} \ind{\text{$\hmap(u)$ and $\hmap(v)$ are separated by $e$}}\,b(f)\\
&= \sum_{f=uv \in E(T)} \ind{\text{$\oriented{F}{u}$ and $\oriented{F}{v}$ orient $e$ in opposite directions}}\,b(f)\\
&= \sum_{f=uv \in E(T)} \ind{\text{exactly one of $u$ and $v$ is in $S_e$}}\,b(f)\\
&= b(\delta(S_e)).
\end{align*}

Combined with \eqref{eq:ucut}, this completes the proof.
\end{proof}

\paragraph*{General networks}
We now show how this result for tree networks can be leveraged to demonstrate the general case. 

\medskip

Let $(\phi, u)$ be any $F$-hubbing for $\treeUb{T}{b}$ into $G$.
By Theorem~\ref{thm:opthhtree}, there exists a $T$-hubbing $(\eta, q^*)$ for $\treeUb{T}{b}$ into $F$, where $q^*$ is the minimal possible capacity allocation discussed earlier.
We will essentially compose this hierarchical hubbing with $\phi$ to obtain a $T$-hubbing $(\rho, u)$ into $G$. 
    We define the $T$-embedding $\rho$ as follows.
    Let $\rho(v) = \phi(\eta(v))$ for all $v \in V(T)$. 
    For any edge $vw \in E(T)$, consider the path $\eta(vw)$ in $F$, and write it in terms of its edges: $\eta(vw) = e_1e_2\ldots e_t$.
  Take $\rho(vw)$ to be any simple $\rho(v)$-$\rho(w)$ path in $G$ contained in the concatenation of the paths $\phi(e_1)$, $\phi(e_2), \ldots, \phi(e_t)$.
   Clearly $\rho$ does define a $T$-embedding into $G$.

    We have that for any $e \in E(G)$, 
    \[ u(e) \geq \sum_{f \in E(F)} \ind{e \in \phi(f)}\;q^*(f). \]
    This follows from the definition of a $F$-hubbing for $\treeUb{T}{b}$, combined with~\eqref{eq:q}.
    Now since $(\eta, q^*)$ is a $T$-hubbing for $\treeUb{T}{b}$ into $F$, we have that
    \[ q^*(f) \geq \sum_{f' \in E(T)} \ind{f' \in \eta(f)}\; b(f'). \]
    Hence
    \begin{align*}
       u(e) &\geq \sum_{f' \in E(T)}\sum_{f \in E(F)} \ind{e \in \phi(f)} \cdot \ind{f' \in \eta(f)}\;b(f')\\
      &\geq \sum_{f' \in E(T)} \ind{e \in \rho(f')}\;b(f'). \end{align*}
  So $(\rho, u)$ is indeed a $T$-hubbing for $\treeUb{T}{b}$ into $G$. The proof of Theorem~\ref{thm:opthh}, and hence Theorem~\ref{thm:main}, is complete.

\section{Connection to the VPN Conjecture}\label{sec:vpnconj}
A well-known conjecture in the area was the \emph{VPN Conjecture}. This was resolved by Goyal et al.~\cite{GOS13}, who proved
\begin{theorem}[\cite{GOS13}]
    There is an optimal solution to the RND problem with oblivious routing for the hose model whose support is a tree.
\end{theorem}
This result leads to a polynomial time algorithm to solve the RND problem for the hose model exactly: indeed, Gupta et al.~\cite{guptavpn01} had previously provided a polynomial time algorithm that computes the optimal tree solution.
The algorithm simply finds a hub routing solution. In other words, for each node $v$, compute the sum (weighted by the hose capacities $b_i$) of the lengths of the shortest paths from each terminal to $v$; then choosing $v^*$ that minimizes this quantity, allocate $b_i$ units of capacity along the $i$-$v^*$ shortest path for each $i \in W$, additively\footnote{A technicality: the marginals should be ``defining'' in the sense of Definition~\ref{def:defining}; here, this means that $b_i \leq \tfrac12\sum_j b_j$ for all $i \in W$.}. 
The routing is is simply a hub routing centered at $v^*$: for each pair $i,j$ of terminals, the path from $i$ to $j$ is obtained by appending the shortest path from $i$ to $v^*$ to that from $v^*$ to $j$.
This is of course a $R^{b'}$-\hh{} solution, where $R$ is a star with leaf set $W$, and the capacity of edge $ir$ (with $r$ being the internal node of the star) is $b'(ir) = b_i$.

The result of Gupta et al.~\cite{guptavpn01} can be seen as a precursor to, and evidence for, the VPN Conjecture.
Similarly, the result of this paper, which is a generalization of the result of Gupta et al., can be seen as evidence for the following conjecture proposed in~\cite{OS2010}.
\begin{gvpnconjecture}
    There is an optimal solution to the RND problem with universe $\treeUb{T}{b}$ that is a $T$-hubbing solution.
\end{gvpnconjecture}
Theorem~\ref{thm:main} shows that the following version of the conjecture, while seemingly weaker, is equivalent:
\begin{gvpnconjecturetwo}
    There is an optimal solution to the RND problem with universe $\treeUb{T}{b}$ that is a hierarchical hubbing solution.
\end{gvpnconjecturetwo}
It is known that the optimal $T$-hubbing solution for $\treeUb{T}{b}$ is always within a constant factor of the optimal oblivious routing solution~\cite{OS2010}.
This in particular implies that the RND problem for the generalized hose model is constant approximable.
A positive resolution to the Generalized VPN Conjecture would imply that the RND problem for $\treeUb{T}{b}$ can be solved optimally in polynomial time.

\paragraph{Acknowledgements} The author is very grateful to Alexandre Fr\'echette and Bruce Shepherd, both for asking the question that led to this note, and for their help in the exposition of the main proof.

\bibliographystyle{abbrv}
\bibliography{rnd_all}

\end{document}